\documentclass{article}
\usepackage{geometry,setspace,times,ifthen}
\usepackage{amsmath,amsfonts,amssymb,paralist,latexsym,mathtools,xparse,natbib, epsfig,color,rotating,times,float,subcaption,algorithm,verbatim,enumitem,bm,graphicx,listings,url,multirow}
\usepackage{xifthen,ifthen}%
\usepackage[title]{appendix}
\usepackage{tikz} \usetikzlibrary{patterns,hobby}
\usepackage{pgfplots}

\newcommand{\state}{x}
\newcommand{\bstate}{\bar{\state}}
\newcommand{\statedim}{X}
\newcommand{\action}{a}
\newcommand{\baction}{\bar{\action}}
\newcommand{\tp}{P} \newcommand{\valueb}{V}
\newcommand{\gs}{\geq_{s}} \newcommand{\ls}{\leq_{s}}
\newcommand{\statespace}{\mathcal{X}}
\newcommand{\actionspace}{\mathcal{A}} \newcommand{\actiondim}{A}
\newcommand{\horizon}{N} \newcommand{\discount}{\rho}
\newcommand{\cost}{c} \newcommand{\reward}{r}
\newcommand{\policy}{\mu} \newcommand{\optpolicy}{\policy^*}

\newcommand{\Q}{Q}
\newcommand{\argmax}{\operatornamewithlimits{arg\,max}}
\newcommand{\argmin}{\operatornamewithlimits{arg\,min}}
\newcommand{\E} {\Bbb{E}}

\makeatletter
\newcommand{\al}{
  \alpha
  \@ifnextchar\bgroup{\sb}{}
}
\newcommand{\be}{
  \beta
  \@ifnextchar\bgroup{\sb}{}
}
\newcommand{\ga}{
  \gamma
  \@ifnextchar\bgroup{\sb}{}
}
\makeatother

\newtheorem{theorem}            {Theorem}
\newtheorem{lemma}            {Lemma}

\newtheorem{corollary}            {Corollary}
\newtheorem{proposition}            {Proposition}

\newcounter{assum_index}
\newcounter{ex1_index}

\newcommand{\ID}{\mathcal{I}}
\newenvironment{proof}{\noindent {\bf Proof:}}{\hfill$\square$}
\newcommand{\fun}{\phi}
\newcommand{\reals}{\mathbb{R}}
\newcommand{\prob}{\mathbb{P}}

\newcommand{\ep}{\epsilon}

\newcommand{\valuef}{V}

\NewDocumentCommand{\terminal}{e{_}ge{_}}{%
   {\tau}%
   \IfValueT{#2}{_{#2}}%
 }
 
\NewDocumentCommand{\terminalc}{e{_}ge{_}}{%
   {\tau}%
   \IfValueT{#2}{_{#2}}%
 }

\NewDocumentCommand{\da}{e{_}ge{_}}{%
   {d}%
   \IfValueT{#2}{_{#2}}%
 }

\newcommand{\p}{\prime}

\newcommand{\W}{W}

\newcommand{\fv}{f}
\newcommand{\bcost}{\bar{\cost}}

\newcommand{\pair}{\state,\bstate,\action}

\def\g1{\geq_{r_1}}

\newcommand{\param}[1]{p_{#1}}

\newcommand{\Qw}{\bar{\Q}}

\newcommand{\ain}{\in}
\newcommand{\pp}{p}
\newcommand{\mup}{\Delta}

\newcommand{\tpe}{\tp^\saram}
\newcommand{\saram}{\epsilon}
\newcommand{\fc}{f}
\newcommand{\dparam}{\gamma_\action}

\newcommand{\cactionspace}{\bar{\mathcal{A}}}

\newcommand{\tID}{$\ID$ }
\newcommand{\values}{\valuef^\saram}

\newcommand{\diffcost}{\Delta}
\newcommand{\increasing}{\uparrow}
\newcommand{\bgam}{\bar{\gamma}}

\newcommand{\gc}{>_c}

\newcommand{\gtwo}{>_2}
\newcommand{\nn}{\nonumber}
\newcommand{\tht}{\theta}
\allowdisplaybreaks
\title{Interval Dominance based  Structural Results for Markov Decision Process} 
\author{Vikram Krishnamurthy,  {\tt vikramk@cornell.edu} \\ School of Electrical and Computer Engineering, \\ Cornell University, Ithaca, NY}    
\begin{document}

\maketitle

\begin{abstract}
   Structural results impose sufficient conditions on the model parameters of a  Markov decision process (MDP) so  that the optimal policy is an increasing function of the underlying state.
  The  classical  assumptions  for  MDP  structural results require supermodularity of the rewards and transition  probabilities.
However, supermodularity does  not hold in many applications.
This paper uses a sufficient condition for interval dominance (called \tID) proposed in the micro-economics literature,   to obtain structural results for MDPs under more general conditions. We   present  several MDP examples where  supermodularity does not hold, yet \tID  holds, and so  the optimal policy is monotone; these include sigmoidal rewards (arising in prospect theory for human decision making), bi-diagonal and  
perturbed bi-diagonal transition   matrices (in optimal allocation problems).
We also consider MDPs with TP3 transition matrices and concave value functions.
Finally, reinforcement learning  algorithms that exploit the differential sparse structure of the optimal monotone policy are discussed. 
\end{abstract}


\section{Introduction}
Markov decision processes (MDPs) are controlled Markov chains. Brute force numerical solution to compute the optimal policy  of an MDP  with a large state and action space is expensive and yields little insight into the structure of the controller.
{\em Structural results} for MDPs are widely studied in stochastic control, operations research and economics~\citep{Top98,Ami05,Put94,HS84}. They  impose  sufficient conditions on the parameters of an MDP  model
so that there exists an optimal policy $\optpolicy(\state)$ that is increasing\footnote{We use increasing in the weak sense to mean non-decreasing.}  in the state $\state$, denoted as $\optpolicy(\state) \uparrow \state$.
Such monotone  optimal policies are useful as they  yield insight into the structure of the optimal controller of the MDP.
Put simply, they provide a mathematical justification for rule of thumb heuristics such as choose a ``larger'' control action for a ``larger'' state.  
Also, since monotone optimal policies are differentially sparse (see Sec.\ref{sec:conclusions}),
optimization algorithms and reinforcement learning algorithms that exploit this sparsity can  solve the MDP
efficiently~\citep{Kri16,MRK17}.

The  textbook proof \citep{Put94,HS84} for the existence of a  monotone policy in  a MDP relies on  supermodularity. 
By  imposing sufficient conditions  on
the rewards and transition probabilities of the MDP,  the classical proof shows that  the $\Q$ function  in Bellman's dynamic programming equation is supermodular.
(These conditions are reviewed  in Section \ref{sec:background}.) 
With $\statespace, \actionspace$ denoting a finite state space and action space, recall \citet{Top98}  that a generic  function $\fun: \statespace \times \actionspace \rightarrow \reals $ is supermodular\footnote{More generally supermodularity  applies to lattices with a partial order \citep{Top98}. In our  simple setup of~\eqref{eq:supermod}, \citet{Put94} uses the terminology `superadditive' instead of supermodular.} if it has increasing differences:
\begin{equation}
  \label{eq:supermod}
   \fun(\bstate,\baction) - \fun(\bstate,\action)  \geq 
   \fun(\state,\baction) - \fun(\state,\action) ,
   \quad\bstate > \state,\;\baction > \action.
\end{equation}
Topkis' theorem then states that supermodularity is a sufficient condition for
\begin{equation}
  \action^*(\state) \ain \argmax_{\action \in \actionspace} \fun(\state,\action) \uparrow \state. \label{eq:argmax}
\end{equation}
So if it can be shown for an MDP that its $Q$ function is supermodular, then Topkis theorem  implies that there exists an optimal policy that  is monotone:
$$ \policy^{*}(\state)\ain \argmax_{\action \in \actionspace} \Q(\state,\action)
\uparrow \state $$

However, supermodularity can be a restrictive  sufficient condition for the existence of a monotone optimal policy; it imposes conditions on the
rewards and transition probabilities that may  not hold in many cases.

More recently,   \citet{QS09} introduced the  \textit{Interval Dominance} condition    which is necessary and sufficient for \eqref{eq:argmax} to hold.  For the purposes of our  paper, ~\citet[Proposition~3]{QS09} gives the following  useful sufficient condition\footnote{If $\al{\pair}$ is a fixed constant independent of $\pair$, then \eqref{eq:id} is  sufficient for the single crossing property~\citep{MS94}, namely, RHS of~\eqref{eq:supermod} $\geq0$ implies
  LHS of~\eqref{eq:supermod} $\geq0$. Supermodularity implies single crossing
  which in turn implies interval dominance; see also~\citet{Ami05} for a tutorial exposition. The condition~\eqref{eq:id} is sufficient for interval dominance and is the main condition that we will use in this paper.}
  for
  $\fun: \statespace \times \actionspace \rightarrow \reals$ to satisfy  interval dominance:
\begin{equation}
  \label{eq:id}
 \fun(\bstate,\action+1) - \fun(\bstate,\action)  \geq \al{\pair} \big[
  \fun(\state,\action+1) - \fun(\state,\action) ], \quad \bstate > \state
\end{equation}
where the  scalar valued function $\al{\pair} > 0$ (strictly non-negative) is increasing in $\action$.
We symbolically  denote~\eqref{eq:id} as  the  condition
  $(\fun,\al) \in \ID$.
Comparing 
supermodularity \eqref{eq:supermod} with 
 \tID  in~\eqref{eq:id}, we see that supermodularity is a special case of~\tID  when $\al{\pair} = 1$.
An important property  of  \tID is that it compares adjacent actions $a$ and $a+1$. A
more restrictive 
 condition would be to replace $a+1$ with any action $\bar{\action} > \action$  in~\eqref{eq:id}. However, this stronger condition (which in analogy to~\eqref{eq:supermod} can  be called $\alpha$-supermodularity) is highly restrictive and does not hold for MDP examples considered below.
\\

{\bf Main Results}.
This paper   shows how \tID in~\eqref{eq:id} applies to obtain structural results for MDPs under more general conditions than the textbook supermodularity conditions. 
Theorems~\ref{thm:IDmonotone} and~\ref{thm:convexdom} are our main results. To avoid technicalities we  consider finite state, finite action MDPs which are either finite horizon or discounted reward infinite horizon.
We   present  several MDP examples where the $\Q$ functions satisfies \tID but not supermodularity, and  the optimal policy is monotone. One important class
comprises MDPs with sigmoidal and concave rewards; since a sigmoidal function comprises convex and concave segments, supermodularity rarely holds. Such sigmoidal rewards arise in prospect theory (behavioral economics) based models for human decision making \citep{KT79}.
A second important class of  examples  we will consider involves perturbed bi-diagonal transition   matrices for which the standard supermodularity assumptions do not hold.
Bi-diagonal transition matrices arise in optimal allocation with penalty costs~\citep{DLR76,Ros83}.
The  result in Sec.~\ref{sec:ross}  complements this  classical
result  for possibly non-submodular costs.
Finally, a third class of examples
comprises  MDPs with integer concave value functions.  Theorem~\ref{thm:convexdom}  and Corollary~\ref{cor:convexdom} impose  TP3 Theorem~\ref{thm:convexdom}  and Corollary~\ref{cor:convexdom} impose  TP3 (totally positive of order 3) assumptions along with \tID to show that the optimal policy is monotone. An  extension of the classical
TP3 result of~\citet[pg 23]{Kar68} is proved  to characterize  the \tID condition for MDPs with bi-diagonal and tri-diagonal transition matrices.  Such MDPs model  controlled random walks \citep{Put94} and arise in the control of queuing and manufacturing systems. 

\section{Background.  Supermodularity based   Results} \label{sec:background}
An infinite horizon discounted reward MDP model is the tuple
$  (\statespace,\actionspace, (\tp(\action), \reward(\action),\action \in \actionspace),\discount)$.
Here $\statespace = \{1,\ldots,\statedim\}$ denotes the finite state space, and
we will denote $\state_k \in \statespace$ as the state at time $k=0,1,\ldots$.
Also
$\actionspace = \{1,\ldots,\actiondim\}$ is the action space, and we will denote $\action_k \in \actionspace$ as the action chosen at time $k$. $\tp(\action)$ are $\statedim\times \statedim$ stochastic matrices with elements $\tp_{ij}(\action) = \prob(\state_{k+1}=j|\state_k=i,\action_k=\action)$,  $\reward(\action)$ are $\statedim$ dimensional reward vectors with elements denoted $\reward(\state,\action)$, and $\discount \in (0,1)$ is the  discount factor.

The action  at each time $k$ is chosen as
$\action_k = \policy(\state_k)$ where  $\policy$ denotes a stationary  policy $\policy: \statespace\rightarrow \actionspace$. 
The optimal stationary policy  $\optpolicy:\statespace\rightarrow \actionspace$
is the maximizer of the infinite horizon discounted reward $J_\policy$:  
\begin{equation}
  \label{eq:mdpcost}
  \optpolicy(\state)  \ain \argmax_\policy J_\policy(\state), \quad
  J_\policy(\state) = \E_\policy \{\sum_{k=0}^\infty \discount^k \reward(\state_k,\action_k) \mid \state_0 = \state\}
\end{equation}
The optimal stationary policy  $\optpolicy$  satisfies Bellman's dynamic programming  equation
\begin{equation}
  \label{eq:dp}
  \begin{split}
    \optpolicy(\state)  \ain \argmax_{\action \in \actionspace}\{ \Q(\state,\action)\} , \quad
    \valueb(\state) = \max_{\action \in \actionspace}\{ \Q(\state,\action)\},  \quad
    \Q(\state,\action) = 
  \reward(\state,\action) +
  \discount \sum_{j=1}^\statedim \tp_{\state j}(\action) \,\valueb(j) 
\end{split}
\end{equation}
An   MDP with finite horizon $\horizon$ is the tuple
$  (\statespace,\actionspace,(\tp(\action),\reward(\action), \action\in \actionspace), \terminal)$
where $\terminal$ is the $\statedim$-dimensional terminal reward vector. (In general $\tp(\action)$ and $\reward(\action)$ can
depend on time $k$; for notational  convenience  we suppress this time dependency.)
The optimal policy sequence $\policy_0,\ldots,\policy_{\horizon-1}$ is given by Bellman's recursion:
$\valueb_{\horizon}(\state) = \terminal{\state}$, $\state \in \statespace$, and for $ k=0,\ldots,\horizon$,
\begin{equation}
  \label{eq:finitedp}
  \begin{split}
    \optpolicy_k(\state)  \ain  \argmax_{\action \in \actionspace}\{ \Q_k(\state,\action)\} , \;
    \valueb_k(\state) = \max_{\action \in \actionspace}\{ \Q_k(\state,\action)\},  \quad
    \Q_k(\state,\action) = 
  \reward(\state,\action) +
  \sum_{j=1}^\statedim \tp_{\state j}(\action) \,\valueb_{k+1}(j) 
  \end{split}
\end{equation}

\subsubsection*{Monotone Policies  using Supermodularity}

%
%
For an MDP, the   textbook sufficient conditions  for $\Q$ in~\eqref{eq:dp} or ~\eqref{eq:finitedp} to be supermodular are  
\begin{enumerate}[label=(A{\arabic*})]
\item \label{cost}
Rewards  $\reward(\state,\action)$  is increasing in $\state$ for each $\action$.
  
\item \label{tps}
  $\tp_\state(\action) \ls \tp_{\state+1}(\action)$ for each $\state,\action$,
  where $\tp_\state(\action)$ is the $\state$-th row of  matrix $\tp(\action)$.
 \footnote{$\ls$ denotes first order stochastic dominance, namely,
  $\sum_{j=l}^\statedim \tp_{\state,j}(\action) \leq \sum_{j=l}^\statedim \tp_{\state+1,j}(\action)$, $l\in \statespace$. \label{foot:fod}}
\item \label{supermod_cost}
  $\reward(\state,\action)$ is supermodular in $(\state,\action)$.
\item \label{tp_supermod}
  $\sum_{j \geq l} \tp_{\state j}(\action)$ is supermodular in $\state,\action$ for each  $l \in \statespace$.
    \setcounter{assum_index}{\value{enumi}}
\item \label{terminal}
 The terminal reward  $\terminal{\state} \increasing \state$.
  \setcounter{assum_index}{\value{enumi}}
\end{enumerate}


The  following textbook result establishes $Q_k$ and $Q$ are supermodular;  so the optimal policy is monotone:
\begin{proposition}[\citep{Put94,HS84}]
\label{res:textbook}
  (i) For  a discounted reward  MDP, under \ref{cost}-\ref{tp_supermod}, the optimal policy $\optpolicy(\state) $ satisfying \eqref{eq:dp}
  is increasing\footnote{More precisely, there exists a version of the optimal policy that is non-decreasing in $\state$. Recall that \eqref{eq:mdpcost} uses the notation  $\ain$ since the optimal policy is not necessarily unique.} in $\state$.\\
(ii)   For a finite horizon MDP,  under \ref{cost}-\ref{terminal}, the
  optimal policy sequence $\optpolicy_k(\state)$, $k=0,\ldots,\horizon-1$,   satisfying~\eqref{eq:finitedp} is increasing   in $\state$.
\end{proposition}

\section{MDP Structural Results using Interval Dominance} \label{sec:mdp}

The supermodular conditions \ref{supermod_cost},  \ref{tp_supermod} on the rewards and transition probabilities, are restrictive. We 
relax  these with the  interval dominance condition \tID defined in~\eqref{eq:id}
as follows:


\begin{enumerate}[label=(A{\arabic*})]
  \setcounter{enumi}{\value{assum_index}}
\item \label{cost_id} For $\be{\pair} > 0$ and increasing in $\action$, the rewards satisfy
\begin{equation}
\label{eq:2}
\reward(\bstate,\action+1) - \reward(\bstate,\action) \geq
\be{\pair} \, \big[\reward(\state,\action+1) - \reward(\state,\action) \big], \quad \bstate > \state
\end{equation}

\item \label{tp_id} For each $l\in \statespace$, the transition
  probabilities satisfy
  %
  \begin{equation}
    \label{eq:IDTP}
    \sum_{j\geq l} (\tp_{\bstate,j}(\action+1) -  \tp_{\bstate,j}(\action))
    \geq \al{\pair} \biggl[ \sum_{j\geq l} (\tp_{\state,j}(\action+1) -  \tp_{\state,j}(\action))\biggr], \quad \bstate > \state
  \end{equation}
  where $\al{\pair} > 0$ and is increasing
  in $\action$. ($\al{\pair}$ is not allowed to depend on $l$.)
  \\
  Equivalently, \eqref{eq:IDTP} can be expressed in terms of first
  order stochastic dominance $\gs$ as
  \begin{equation}
    \label{eq:id2}
    \frac{ \tp_{\bstate}(\action+1) + \al{\pair}\,\tp_\state(\action)}
    {1+\al{\pair}}
    \gs    \frac{ \tp_{\bstate}(\action) + \al{\pair}\,\tp_\state(\action+1)}
    {1+\al{\pair}}, \quad \bstate > \state
  \end{equation}

\item \label{sum}
There exist $\al{\pair} = \be{\pair}$ for which~\ref{cost_id}, \ref{tp_id}
  hold.
  \setcounter{assum_index}{\value{enumi}}
\end{enumerate}

{\em Remark.}  
If $\al{\action}=1$ and $\be{\action}=1$, then
 \ref{cost_id} and \ref{tp_id} are equivalent
  to the supermodularity conditions \ref{supermod_cost} and
  \ref{tp_supermod}. Then ~\ref{sum} holds trivially. Note that 
 \ref{sum} is sufficient for the sum of two \tID functions to be \tID. 



 {\bf Main Result}.
The following is our main result.

\begin{theorem} \label{thm:IDmonotone} (i) For a discounted reward
  MDP, under \ref{cost}, \ref{tps}, \ref{cost_id},
\ref{tp_id}, \ref{sum}, there exists an  optimal stationary  policy $\optpolicy(\state)$ satisfying
  \eqref{eq:dp} which is increasing in $\state$.  (ii) For a finite horizon
  MDP, under \ref{cost}, \ref{tps}, \ref{terminal},
  \ref{cost_id}, \ref{tp_id}, \ref{sum}, there exists an  optimal policy
  sequence $\optpolicy_k(\state)$, $k=0,\ldots,\horizon$
  satisfying~\eqref{eq:finitedp} which is increasing in~$\state$.
\end{theorem}

\textit{Remark.} Theorem~\ref{thm:IDmonotone} also holds for average reward MDPs that are unichain \cite{Put94} so that  a stationary optimal policy  exists. This is because our proof uses the  value iteration algorithm,  and for average reward   problems, the same ideas directly apply to the relative  value iteration algorithm.

\begin{proof}
  The standard textbook proof \citep{Put94} shows via induction that
for the finite horizon case, 
  \ref{cost},
  \ref{tps}, \ref{terminal}  imply that $\Q_{k}(\state,\action)$ is increasing in $\state$
  for each $\action \in \actionspace$, and therefore $\valueb_{k}(\state)$
  is increasing in $\state$.
  The induction step also constitutes the value iteration algorithm for  the infinite horizon case, and shows that $Q(\state, \action)$ and $\valueb(\state)$ are increasing in $\state$.  
  
  Next, since $\valueb(\state)$ is increasing, 
 assumption \eqref{eq:id2} in
  \ref{tp_id} implies that for  $\bstate> \state$,
  \begin{equation}
    \sum_{j=1}^\statedim \big[\tp_{\bstate,j}(\action+1) - \tp_{\bstate,j}(\action)
    \big]\valueb(j ) \geq \al{\pair} \big(
    \sum_{j=1}^\statedim \big[\tp_{\state,j}(\action+1) - \tp_{\state,j}(\action)
    \big]\valueb(j ) \big) \label{eq:step1}
  \end{equation}

    Assumption \ref{cost_id} implies the rewards satisfy \tID.
Finally, 
\ref{sum} implies 
for  $ \bstate> \state$,
\begin{multline} \label{eq:sumstep}
  \reward(\bstate,\action+1) - \reward(\bstate,\action)
   +  \sum_{j=1}^\statedim \big[\tp_{\bstate,j}(\action+1) - \tp_{\bstate,j}(\action)
    \big]\valueb(j )
  \\  \geq \ga{\pair}
 \bigg(  \reward(\state,\action+1) - \reward(\state,\action)
    \sum_{j=1}^\statedim \big[\tp_{\state,j}(\action+1) - \tp_{\state,j}(\action)
    \big]\valueb(j ) \bigg)
  \end{multline}
  for $\gamma=\al = \be$.
  Thus
  $(\Q,\gamma) \in \ID$ implying that \eqref{eq:argmax} holds.
\end{proof}

\subsection{Example 1. MDPs with Interval Dominant Rewards}
\label{sec:ex1}

Our first example considers MDPs with  sigmoidal and concave\footnote{Throughout this paper convex (concave)  means integer convexity (concavity). Since   $x \in \{1,\ldots,\statedim\}$, integer  convex $\fun$  means $\fun(x+1) - \fun(x) \geq \fun(x) - \fun(x-1)$. We do not consider higher dimensional discrete convexity such as multimodularity; see Sec.\ref{sec:conclusions}.} rewards.   Supermodularity is difficult to ensure  since a sigmoidal reward  comprises a convex segment  followed by a concave segment.
In Figure~\ref{fig:ex1a}, reward $\reward(x,1)$ is  sigmoidal, while  $\reward(x,2)$ and $ \reward(x,3)$ are concave in $\state$.
Since concave reward $\reward(x,3)$ intersects sigmoidal reward  $\reward(x,1)$ multiple times, the  single crossing condition and therefore supermodularity ~\ref{supermod_cost} does not hold.  Also $r(x,3) -r(x,1)$ is not increasing and so  not supermodular.  But  condition \tID \ref{cost_id} holds. Specifically, $r(x,2) - r(x,1)$ is single crossing, and $r(x,3) -  r(x,2)$ is  single crossing.
Note that \tID does not require  $r(x,3) - r(x,1)$ to be single crossing.

\begin{figure}[h] \centering
  \begin{subfigure}{0.45\linewidth}
    \includegraphics[scale=0.45]{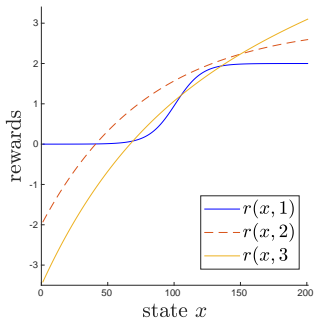}
   \caption{Rewards} \label{fig:ex1a}
 \end{subfigure}
 \begin{subfigure}{0.45\linewidth}
    \includegraphics[scale=0.4]{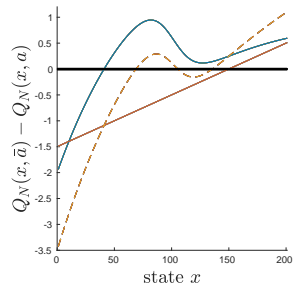}
    \caption{$Q$-function for MDP}
    \label{exi:sigmoidreward}
  \end{subfigure}
    \caption{Interval Dominant Rewards that are not  single crossing and so  not supermodular.  If supermodularity holds then  the curves would be increasing with $\state$. Yet  \tID holds by Corollary~\ref{cor:ex0} and the optimal policy is monotone; see Example (i).}
  \label{fig:ex1rewards}
\end{figure}

    Consider a discounted reward MDP. Assume:
    \begin{enumerate}[label=(Ex1.{\arabic*})]
    \setcounter{enumi}{\value{ex1_index}}
 
\item \label{c1:r3}
For each pair of actions $\action,\action+1$, assume there is state $x_a^*$ such that
$\reward(x,a+1) \leq  \reward(x,a)$,   $\tp_x(a+1) \ls \tp_x(a)$ for $x \leq x_a^*$. Also  $\reward(x,a+1) \geq \reward(x,a)$, $\tp_{x}(a+1) \gs \tp_x(a)$   for
$x \geq  x_a^*$.
\end{enumerate}

\begin{corollary} \label{cor:ex0}
  Consider a discounted reward MDP.
   Assume
  \ref{cost}, \ref{tps}, \ref{c1:r3}. 
  Then Theorem~\ref{thm:IDmonotone} holds.\end{corollary}
 
Compared to textbook Proposition~\ref{res:textbook},  Statement 1 of
Corollary \ref{cor:ex0} does not impose supermodularity conditions on the rewards
or transition probabilities.  
\ref{c1:r3}  is  weaker than the single crossing condition.

\begin{proof}
We verify that the conditions \ref{cost_id}, \ref{tp_id},  ~\ref{sum} of Theorem~\ref{thm:IDmonotone} hold:


First consider  $x < \bstate \leq x^*_a$. Since $\reward(x,a) \geq \reward(x,a+1)$, and  $\reward(\bstate,a) \geq \reward(\bstate,a+1)$,
\ref{cost_id} holds for
all  $\be \in [\be^*_{\pair},\infty)$ for some $\be^*_{\pair} > 0$.
Also $\tp_{x}(\action+1) \ls \tp_x(a)$
implies~\ref{tp_id} holds for all $\al{\pair} \in [\al^*_{\pair},\infty)$ for some $\al^*_{\pair} > 0$.
So  we can choose $\al = \be=\max_a\{\al^*_{\pair},\be^*_{\pair}\}$ independent of $a$ so that~\ref{sum} holds.

Next  consider $ \bstate > x\geq   x^*_a$.  Then \ref{cost_id} holds for
all  $\be \in (0,\be^*_{\pair}]$ for some $\be^*_{\pair} > 0$.
Also $\tp_{x}(\action+1) \gs \tp_x(a)$
implies~\ref{tp_id} holds for all $\al \in (0,\al^*_{\pair}]$ for some $\al^*_{\pair} > 0$.
Therefore, we can choose $\al = \be = \min_a\{\al^*_{\pair},\be^*_{\pair}\}$ independent of $a$ so that~\ref{sum} holds.
Finally, for $x \leq x^*_a$ and $\bstate > \state^*_a$, \ref{cost_id} and~\ref{tp_id} hold
for all $\al, \be > 0$. 
So  Theorem \ref{thm:IDmonotone} applies and  $\policy^{*}(\state) \uparrow \state$.
\end{proof}

\subsubsection*{Example (i). Sigmoidal\footnote{Sigmoidal rewards/costs are ubiquitous. They  arise in logistic regression, prospect theory in  behavioral economics, and  wireless communications.}  and Concave Rewards}  The following MDP parameters  satisfy the assumptions of Corollary~\ref{cor:ex0}:
$\statedim=201$, $\actiondim=3$.
The action dependent transition matrices are 
$$\tp_i(1) = \tp_{i-1}(1) + \mu (e_\statedim- e_1),
\quad \tp_i(a+1) = \begin{cases}  \tp_i(a) - \ep (e_\statedim-e_1), \; i \leq 50,  \\
                     \tp_i(a) +  \ep (e_\statedim-e_1), \; i > 50
                   \end{cases}                 
\mu = \frac{0.004}{\statedim},\; \ep = \frac{0.05}{\statedim} $$
Here  $e_i$ denotes the unit $\statedim$-dimension row vector with 1 in the $i$-th position.

The rewards parametrized by $\tht=[2,\statedim-1,20, 5,80,-2,5,80,-3.5,0.01]$ are
\begin{equation} 
  \begin{split}  
\reward(\state,1) &= \frac{\tht_1}{1 + \exp(\frac{\state - \tht_2}{\tht_3})}, \;
\reward(\state,2) = \tht_4 ( 1 - \exp(-\frac{\state}{\tht_5}) ) + \tht_6, \\
\reward(\state,3) &= \tht_7 ( 1 - \exp(-\frac{\state}{\tht_8}) ) + \tht_9 + \tht_{10}\,\state,
\end{split}
\label{eq:ex1rewards}
\end{equation}
 Figure~\ref{exi:sigmoidreward} shows the non-supermodular $Q_\horizon$
for $\horizon=100$, $\discount=0.9$.
$Q_N(\state,3) - Q_N(\state,1)$ (broken line)  intersects the horizontal axis three times; so single crossing does not hold.   $Q_N(x,2) - Q_N(x,1)$ (blue line) is non-monotone (non-supermodular).   Statement 1, Corollary~\ref{cor:ex0} applies; so the optimal policy is monotone.

\subsubsection*{Example (ii). Prospect Theory based rewards}
In prospect theory \citep{KT79},  an agent (human decision maker)~$a$ has utility  $\reward(x,a)$ that is  asymmetric sigmoidal in $x$. This asymmetry reflects a human decision  maker's risk seeking behavior (larger slope)  for losses and risk averse behavior (smaller slope) for gains.
 With $\statedim$  an  even integer, the prospect theory rewards are
\begin{equation}
  \label{eq:prospect}
r(x,a) = \frac{2 (\mu\,(x-1))^{\tht(a)}}{1 +( \mu\, (x-1))^{\tht(a)}} -1 ,\quad
 \mu = 2/(\statedim-2),\;\tht(a) > 1 
\end{equation}
so  they cross zero at $x=\statedim/2$. The shape parameter $\tht(a)$ determines the slope of the  reward  curve $r(x,a)$.

Suppose the  agents (investment managers)  range from  $a=1$ (cautious) to  $a=\actiondim$ (aggressive); so the shape parameter $\tht(a) \uparrow a$.
The value of an investment evolves according to   Markov chain $\state_k$  with transition probabilities $\tp(a_k)$ based on   agent $a_k$. Since agent $a+1$ is more aggressive (risk seeking) than agent $a$ in losses and gains,  it incurs higher volatility. So the $x$-th row of $\tp(a)$ and $\tp(a+1)$ satisfy 
\begin{equation}
  \label{eq:volatility}
\tp_{x}(\action+1) \ls \tp_x(a), \;
x< \statedim/2 \quad\text{  and }  \tp_{x}(\action+1) \gs \tp_x(a), \; x \geq \statedim/2
\end{equation}
The aim is to choose the optimal agent $a_k$ at each time  $k$ to maximize the discounted infinite horizon reward.  Since  $\reward(x,a)$ is single crossing but not supermodular, \ref{supermod_cost} does not apply.

 \begin{corollary} \label{cor:prospect} Consider a discounted reward MDP with $r(x,a)$ specified by~\eqref{eq:prospect} and $\tht(\action) \uparrow a$.
   Assume~\ref{tps},  \eqref{eq:volatility} hold.
   Then Theorem~\ref{thm:IDmonotone} holds.
 \end{corollary}
The proof  follows from  Corollary~\ref{cor:ex0}
 with $x^*_a = X/2$.


\subsection{Example 2. Interval Dominant Transition Probabilities}

\begin{corollary} \label{cor:ex2} Consider the discounted reward MDP
  with  $\reward(\state,\action)=\fun(\state)$ where $\fun$ is increasing and non-negative in $x$.
  Suppose the $i$-th row of transition  matrix  $\tp(\action)$ is 
\begin{equation}
\label{eq:3}
\tp_{i}(\action)= \pp + \mup_{i,\action}\,(e_\statedim-e_1)
\end{equation}
Here  $e_i$ denotes the unit $\statedim$-dimension row vector with 1 in the $i$-th position.
 $p$ is an arbitrary $\statedim$-dimensional probability row vector.
 Also
 $\mup_{1,a}=0$, $\mup_{i,\action} \in [0,1]$ are increasing in $i$, and satisfy \tID~\eqref{eq:id}.
 (Also, $\mup_{i,\action}\leq \min\{\pp_1, 1 - \pp_\statedim\}$ to ensure $\tp(\action)$ is valid transition matrix.)
  Then 
 Theorem~\ref{thm:IDmonotone} holds.
\end{corollary}

Compared to supermodularity~\ref{tp_supermod} of the transition probabilities, Corollary~\ref{cor:ex2} imposes  weaker conditions:  
$\mup$ satisfy \tID~\eqref{eq:id} and $\pp$ can be any probability vector. Since $\mup$ only needs to satisfy \tID (suitably scaled and shifted to ensure valid probabilities), \eqref{eq:3} offers considerable flexibility in choice of  the transition  matrices.

\begin{proof}
Reward $\reward(\state,\action) = \fun(\state)$ satisfies~\ref{cost}, \ref{cost_id} for all $\be{\pair} > 0$.
  Also $\mup_{\state,\action} \uparrow \state$ implies~\ref{tps} holds.
  Next let us verify~\ref{tp_id}. Using~\eqref{eq:3}, we need to verify 
\begin{equation}
\label{eq:mupcondition}
  (\mup_{\bstate,\action+1}-\mup_{\bstate,\action})\, \sum_{j\geq l} (e_\statedim - e_1)^\p e_j  \geq \al{\pair}
  \big[  (\mup_{\state,\action+1}-\mup_{\state,\action})\, \sum_{j\geq l} (e_\statedim - e_1)^\p e_j \big]
\end{equation}
where $\al{\pair}>0$ is increasing in $\action$.  Since $ \sum_{j\geq l} (e_\statedim - e_1)^\p e_j \geq 0$, clearly $\mup_{i,a}$ satisfying~\eqref{eq:id} for some $\al{\pair} > 0$ increasing in $\action$ is a sufficient condition for~\eqref{eq:mupcondition} to hold. 
%
  Since the choice of $\be{\pair}>0$ is unrestricted, we can choose
  $\be{\pair} = \al{\pair}$. Hence~\ref{sum} holds. Thus Theorem~\ref{thm:IDmonotone} holds.
\end{proof}

   {\bf Example}.  Suppose
$\pp$ is an arbitrary probability vector, and $\mup$ is chosen as the rewards~\eqref{eq:ex1rewards} suitably scaled and shifted. Then the transition matrices inherit the sigmoidal and concave structures of 
Sec.~\ref{sec:ex1}.

\subsection{Example 3. Discounted MDP with Perturbed Bi-diagonal Transition Matrices}
\label{sec:discbi}
This section illustrates the \tID condition in  MDPs with perturbed bi-diagonal transition matrices. 
The E-companion discusses an example in optimal allocation  problems  with penalty costs \citep{Ros83,DLR76}.  It also has  applications in wireless transmission control~\citep{NK10}.

Consider an infinite horizon discounted reward MDP.
The action-dependent   transition matrices
 $\tpe(\action), \action \in \actionspace$ specified by
 parameter $\param{\action} \in [0,1]$ are 
 \begin{equation}
    \begin{split}
      \tpe_{11}(\action)&=1 -(\actiondim-\action)\,\saram , \quad  \tpe_{1,\statedim}(\action) =(\actiondim-\action) \,\saram,   \quad    \tpe_{\statedim,\statedim-1}(\action) = \param{\action}, \quad
        \tpe_{\statedim,\statedim}(\action) =1- \param{\action} \\
      \tpe_{ii}(\action) &=  1-\param{\action}- (\actiondim-\action)\,\saram,\quad \tpe_{i,i-1}(\action)= \param{\action} , \quad \tpe_{i,\statedim}(\action) =(\actiondim-\action) \saram, \quad i= 2,\ldots,\statedim-1 
     \end{split}
    \label{eq:bidiagonale}
  \end{equation}
  where $\saram\ll1$ is a small positive real.
  We assume that $\param{\action}$ is increasing in $\action$.
  When $\saram = 0$, $\tpe(a)$ are bi-diagonal transition matrices;  so $\saram$ can be viewed as a perturbation probability of a bi-diagonal transition matrix. 

Supermodularity    \ref{tp_supermod}  of the transition matrices~\eqref{eq:bidiagonale}  holds if  $\saram \geq \param{\action+1}-  \param{\action}$. 
  In this section we  assume $\saram$ is a small parameter with   $\saram
  \leq \min_a \param{\action+1}- \param{\action}$, so that  \ref{tp_supermod}  does not  hold. Therefore,  textbook Proposition~\ref{res:textbook} does not hold. We  show
  how the \tID condition and Theorem~\ref{thm:IDmonotone} apply.

  {\em Remark.} In our result below, to show condition \tID holds, we choose $\al{\action}  = \be{\action} =( \param{\action+1} - \param{\action})/\saram = \ga{\action}$.
If $\param{\action}$ is differentiable wrt $\action$, then 
as $\saram \rightarrow 0$, i.e., for an MDP with bi-diagonal transition matrices, this can be interpreted as
choosing $\al{\action}  = \be{\action} = d\param{a}/da$.


\begin{corollary} Consider a discounted cost  MDP with transition probabilities~\eqref{eq:bidiagonale}.
  Assume $\param{\action}$ is increasing in $\action$ and $\param{\action+1}-\param{\action} = \ga{a} \saram$ for some positive real number $\ga{a}$ increasing in $\action$.  Assume~\ref{cost} and that 
  \begin{equation}
    \label{eq:cor-bidiag}
    \reward(i+1,a+1) - \reward(i+1,a) \geq \be{\action} \,[\reward(i,a+1) - \reward(i,a)]
  \end{equation}
  for some $\be{\action}$ increasing in $\action$ with $\be{\action} \geq \ga{\action}$.
  Then  optimal policy $\policy^*(\state) \uparrow \state$.
  \label{cor:bi1}
\end{corollary}

\begin{proof}
  We  verify that the assumptions in Theorem~\ref{thm:IDmonotone} hold.
\ref{cost} holds by assumption.
  From the structure of $\tpe(\action)$ in~\eqref{eq:bidiagonale} it is clear that~\ref{tps} holds.
Considering actions $\action$ and $\action+1$,
it is  verified that~\ref{tp_id} holds for all $\al{\action} \geq (\param{a+1}-\param{a})/\saram = \ga{a}$.
Next by assumption~\eqref{eq:cor-bidiag}, \ref{cost_id} holds
for $\be{\action}\geq  \ga{\action}$.
Finally, we can choose $
\al{\action}=\be{\action} \geq \ga{\action}$, and so  \ref{sum} holds. 
\end{proof}

\noindent {\bf Example}.  $\actiondim=2, \statedim=6$,
$\param{1}=0.3$, $\param{2}=\param{1}+20 \saram$, $\saram=10^{-3}$,  $\discount=0.9$, $\horizon=200$,
$\reward^\p = \begin{bmatrix}
  1 & 3.5 & 6 & 6  &  11 & 43 \\
  0  & 2 & 3 & 6 & 12 & 63 
\end{bmatrix}$.
Given the  transition probabilities, we choose $\al\geq 20$.
Also for the rewards,  we  choose $\be=20$ in~\eqref{eq:cor-bidiag}.
So Corollary~\ref{cor:bi1} holds.
Figure~\ref{fig:bidiag1} shows  $Q_\horizon(\state,\action)$ is not supermodular, yet the optimal policy is monotone with $\policy^*(i) = 1$ for $i\in \{1,2,3,4\}$ and $\policy^*(i) = 2$ for $i \in \{5, 6\}$. 

 \begin{figure}
    \centering
    \includegraphics[scale=0.4]{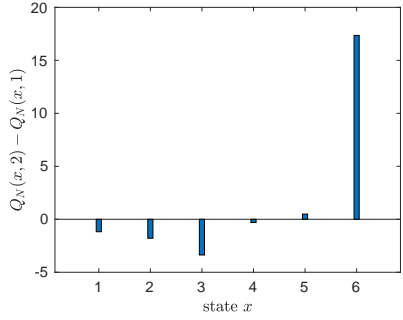}
    \caption{The $Q$-function is  not-supermodular for an MDP with perturbed bi-diagonal matrices; yet the optimal policy $\policy^*(\state)$  is increasing in state  $\state$ by Corollary~\ref{cor:bi1}.}
    \label{fig:bidiag1}
  \end{figure}

\section{Example 4. MDPs with Concave Value Functions} \label{subsec:concave}

Theorem~\ref{thm:IDmonotone}  used  first order dominance and monotone costs to establish \tID and therefore  monotone optimal policies. In comparison,   this section
extends Theorem~\ref{thm:IDmonotone}  to MDPs where the value function is concave. 
We use second order stochastic dominance and concave costs to establish \tID and therefore monotone optimal policies. The results below assume a  TP3  transition matrix;
see~\citet{Kar68} for the rich structure involving their diminishing variation property.
For convenience we minimize costs instead of maximize rewards.

\begin{enumerate}[label=(C{\arabic*})]
\item \label{convexcost}
Costs  $\cost(\state,\action)$ are  increasing and concave in $\state$ for each $\action$.
\item\label{item:TP3}
  $\tp(\action)$ is TP3with $\sum_{j=1}^\statedim j \tp_{ij}(\action)$  increasing and concave in $i$. Totally positive of order 3  means that each 3rd order minor of $\tp(\action)$ is non-negative. 
\item\label{item:submodc}  For  $\be{\pair} > 0$ and increasing in $\action$, 
$\cost(\bstate,\action+1) - \cost(\bstate,\action) \geq
\be{\pair} \, \big[\cost(\state,\action+1) - \cost(\state,\action) \big], \quad \bstate > \state$.
\item\label{item:tpsuperconv}
  For $\al{\pair}> 0$ and increasing in $\action$,
   $ \frac{ \tp_{\bstate}(\action+1) + \al{\pair}\,\tp_\state(\action)}
    {1+\al{\pair}}
    \gtwo    \frac{ \tp_{\bstate}(\action) + \al{\pair}\,\tp_\state(\action+1)}
    {1+\al{\pair}}, \quad \bstate > \state$
  where $\gtwo$ denotes second order stochastic dominance.\footnote{If $p,q$ are probability vectors, then $p\gtwo q$ if  $\sum_{l\leq m} \sum_{j \leq l }p_{j}  \leq \sum_{l\leq m} \sum_{j \leq l } q_j $ for each $m$. Equivalently,
    $p \gtwo q$ iff $f^\p p \geq f^\p q$ for vector $f $ increasing and concave.
    Recall $^\prime$ denotes transpose.}
  \item\label{item:terminalconv} Terminal cost $\terminal{\state}$ is increasing and concave   in $\state$.
  
\end{enumerate}

{\em Remarks}. (i) As shown in the proof, \ref{convexcost} (concavity), \ref{item:TP3},  \ref{item:terminalconv} imply  the value function is concave and increasing.
These together with~\ref{item:submodc},~\ref{item:tpsuperconv} and~\ref{sum} imply  \tID holds and so the optimal policy is monotone.

(ii)
\ref{item:TP3}  generalizes  the assumption that 
$\sum_j j {P}_{ij} $ is linear increasing in~$i$. 
The classical result in~\citet[pg 23]{Kar68} states:
Suppose $P$ is a TP3 transition matrix and  $\sum_j j {P}_{ij} $ is linear increasing in $i$. If vector $\valuef$ is  concave, then vector $P \, \valuef$ is  concave.
However, for bi-diagonal and tri-diagonal transition matrices, $\sum_j j \tp_{ij} $ is concave (or convex) and  not linear in~$i$ (see examples below). This is why we introduced~\ref{item:TP3}. Since the classical result requires   
 $\sum_j j {P}_{ij} $ being  linear in $i$,
it no longer applies. So
we will prove a small generalization that  handles the case
where $\sum_j j \tp_{ij} $ is concave in $i$ (see Lemma~\ref{lem:TP3concave} below).

\begin{theorem} \label{thm:convexdom}
  (i) For a discounted cost
  MDP under \ref{convexcost}-\ref{item:tpsuperconv}, \ref{sum}, optimal policy $\optpolicy(\state) \downarrow \state$.\\  (ii) For a finite horizon
  MDP, under \ref{convexcost}-\ref{item:terminalconv}, \ref{sum}, optimal policy
  sequence $\optpolicy_k(\state)$, $k=0,\ldots,\horizon$
 $\downarrow \state$.
\end{theorem}

\begin{corollary}\label{cor:convexdom}
Consider the modified assumptions: \ref{convexcost}: increasing replaced by decreasing; \ref{item:TP3} concave replaced with convex; \ref{item:submodc}:  inequality involving costs  reversed; \ref{item:tpsuperconv}:  $\gtwo$ replaced by convex dominance\footnote{If $p,q$ are probability vectors, then $p\gc q$ if  $\sum_{l\geq m} \sum_{j \geq l }p_{j}  \geq \sum_{l\geq m} \sum_{j \geq l } q_j $ for each $m$. Equivalently,
    $p \gc q$ iff $f^\p p \geq f^\p q$ for $f $ increasing and convex.} $\gc$; \ref{item:terminalconv}: increasing replaced by decreasing. Under these assumptions and \ref{sum}, Theorem~\ref{thm:convexdom} holds with the modification  $\optpolicy(\state)$ and  $\optpolicy_k(\state)$, are increasing in $\state$.
\end{corollary}

\begin{proof} (Theorem~\ref{thm:convexdom})
 We prove statement (ii). The proof of statement (i) is similar and omitted.

 First we show by induction that $V_k(i)$ is increasing in $i$ for  $k=\horizon,\ldots,1$. By~\ref{item:terminalconv},  $V_\horizon(i) = \terminal{i}$ is increasing. 
Assume $V_{k+1} (i)$ is increasing in $i$.
TP3  assumption~\ref{item:TP3} implies TP2 which preserves monotone functions  \citep[pg 23]{Kar68}, \citep{LC98}, namely,  $\sum_j \tp_{ij}(\action) V_{k+1}(j)$ is increasing in $i$.  This together
 with~\ref{convexcost} implies $Q_k(i,a)$ is increasing. Thus $V_k(i) = \min_a Q_k(i,a)$ is increasing in $i$.
  
  Next  we show by induction that  $V_k(i) $ is concave in $i$.
By~\ref{item:terminalconv}, $V_\horizon= \terminal$ is concave.
Assume $V_{k+1}$ is concave. Then~\ref{item:TP3}
implies $\sum_j \tp_{ij}(\action) \,\valuef_{k+1}(j)$ is concave in $i$ (see Lemma~\ref{lem:TP3concave} below).
Since $\cost(i,a)$ is concave by~\ref{convexcost}, it follows that
 $Q_{k}(i,a) = \cost(i,a) + \sum_j \tp_{ij}(\action) \, \valuef_{k+1}(j) $ is concave in $i$.  Since concavity is preserved by minimization, $V_k(i) = \min_a Q_k(i,a) $ is concave.
 Finally, $V_k(i)$ increasing and concave in $i$  and~\ref{item:tpsuperconv} implies~\eqref{eq:step1} holds for all  $\al{\pair} \geq 1$. Then with~\ref{item:submodc},\ref{sum}, the proof is  identical to~\eqref{eq:sumstep} in Theorem~\ref{thm:IDmonotone}.
\end{proof}

The following lemma used in the proof of   Theorem~\ref{thm:convexdom} 
slightly extends the result in~\citet[pg 23]{Kar68}.

\begin{lemma} \label{lem:TP3concave}
  Suppose $\tp$ satisfies~\ref{item:TP3}.  If $\valuef$ is  concave and increasing, then  $\tp \, \valuef$ is  concave and increasing. 
\end{lemma}

\begin{proof}
First TP3 preserves monotonicity, so $\tp V$ is increasing. Next,
since $\valuef$ is concave and increasing, then for any $a> 0$ and $b \in \reals$,  $V(j) - (aj + b)$ has two or fewer sign changes in the order $-,+,-$ as $j$ increases from 1 to $\statedim$. Let $\fun_i (a,b)= \sum_j \tp_{ij} (aj+b)$. Since $\tp$ is TP3, the diminishing variation property of TP3 implies   $\sum_j \tp_{ij}\, V_j  - \fun_i(a,b)$ also has  two or fewer sign changes in the order $-,+,-$ as $i$ increases from 1 to $\statedim$.
  Assume  two sign changes occur; then for some $i_1 < i_2$,
  $\sum_j \tp_{ij}\, V_j  \geq  \fun_i(a,b)$ for $i_1 \leq i \leq i_2$. Since
  $  \fun_i(a,b)$ is integer  concave in $i$ by~\ref{item:TP3},  it lies above the line segment $L_i$ that connects $(i_1,\fun_{i_1})$ to $(i_2,\fun_{i_2})$.
  So
 $  \sum_j \tp_{ij}\, V_j  \geq \fun_i (a,b)\geq L_i $, $i_1 \leq i \leq i_2$
 Finally, for arbitrary  $i_1<i_2 \in \{1,\ldots,\statedim\}$, we can choose $a = \frac{\sum_j V_j(\tp_{i_2,j} - \tp_{i_1,j})}{\sum_j j (\tp_{i_2,j } - \tp_{i_1,j}) }$ and $b = \sum_j
\tp_{i_1,j}V_j- a \sum_j j\,\tp_{i_1,j}$ so that  $  \sum_j \tp_{ij}\, V_j =  \fun_i (a,b) =  L_i $ at $i=i_1, i_2$. 
Clearly,   $  \sum_j \tp_{ij}\, V_j \geq L_i$ for arbitrary $i_1 \leq i \leq i_2$ and
$  \sum_j \tp_{ij}\, V_j  = L_i$ for $i=i_1,i_2$ implies
$\sum_j \tp_{ij}\, V(j) $ is concave.
\end{proof}

\subsubsection*{Example (i). Bi-diagonal Transition  Matrices and Non-supermodular Costs}

Theorem~\ref{thm:convexdom} applies to bi-diagonal transition matrices with possibly non-supermodular costs; this is in contrast to Sec.~\ref{sec:discbi} where we considered perturbed bi-diagonal matrices.  
Consider an MDP with bi-diagonal transition matrices
$    \tp_{i.i}(a) = 1 - \param{a},  \; \tp_{1,i+1} = \param{a}, \; \tp_{\statedim,\statedim}(\action) = 1, \quad a \in \{1,\ldots,\actiondim\}$.
 Then
$\sum_j j \tp_{ij}(\action) = i + \param{a}$ for $i< \statedim$
and $\statedim$ for $i = \statedim$; so $\sum_j j \tp_{ij}(\action)$ is increasing and concave in $i$ (\ref{item:TP3} holds).
Assume  $\param{\action} \downarrow \action$. Then~\ref{item:tpsuperconv} is equivalent to 
 $\sum_{\l \leq m} \sum_{j \leq l} \tp_{\bstate,j}(a+1) - \tp_{\bstate,j}(a) \leq \al{\pair}
 (\sum_{\l \leq m} \sum_{j \leq l} \tp_{\state,j}(a+1) - \tp_{\state,j}(a)) $.
Since $\param{a} \geq \param{a+1}$, it follows that ~\ref{item:tpsuperconv} holds for all $\al{\pair} \geq 1$. If~\ref{convexcost}, \ref{item:submodc} hold for some
$\be{\pair} > 1$, then  Theorem~\ref{thm:convexdom} holds.

{\em Remarks}. (i) For $\actiondim=2$, and concave increasing costs $\cost(x,1), \cost(x,2)$, the following useful single crossing characterization   satisfies~\ref{item:submodc}.
Suppose $\cost(x=1,2) < \cost(x=1,1)$ and the curves $\cost(x,2)$ and $\cost(x,1)$ intersect once  at $x^*$. For $x \geq x^*$, the curve $c(x,2)$ grows faster than $ c(x,1)$, i.e., $\cost(x,a)$ is supermodular for $x \geq x^*$.
For $x< x^*$, the difference between $\cost(x,1)$ and $\cost(x,2)$ can be arbitrary. Figure~\ref{fig:concavexmdp}(i) illustrates this.
\\
(ii) To motivate Theorem~\ref{thm:convexdom},  \ref{tp_supermod} does not hold for bi-diagonal matrices. Since $\be > 1$,  supermodularity~\ref{supermod_cost} does not hold.   Also,  Theorem~\ref{thm:IDmonotone} does not apply since~\ref{tp_id} does not hold.

{\em Numerical example}. Consider a discounted cost MDP with $\actiondim=2$, $\statedim = 50$,  $\param{1} = 0.8$, $\param{2}=0.7$, 
$\discount = 0.95$,
$\horizon = 200$,
$ \cost(x,1) = \tht_1 x^2 + \tht_2 x + \tht_3, \quad 
\cost(x,2) = \tht_4 \big( 1- \exp(\tht_5 x + \tht_6)\big)  $,
$\tht = [-0.01, 1, 8.8, 25, -0.1, -0.4]$.
It can be verified that the cost is not supermodular, but the conditions of Theorem~\ref{thm:convexdom} are satisfied. So the  value function is concave and  optimal policy is decreasing. Figure~\ref{fig:concavexmdp}(ii)  shows  $\Q_\horizon(\state,\action)$ is not submodular.

\begin{figure} \centering
   \includegraphics[scale=0.4]{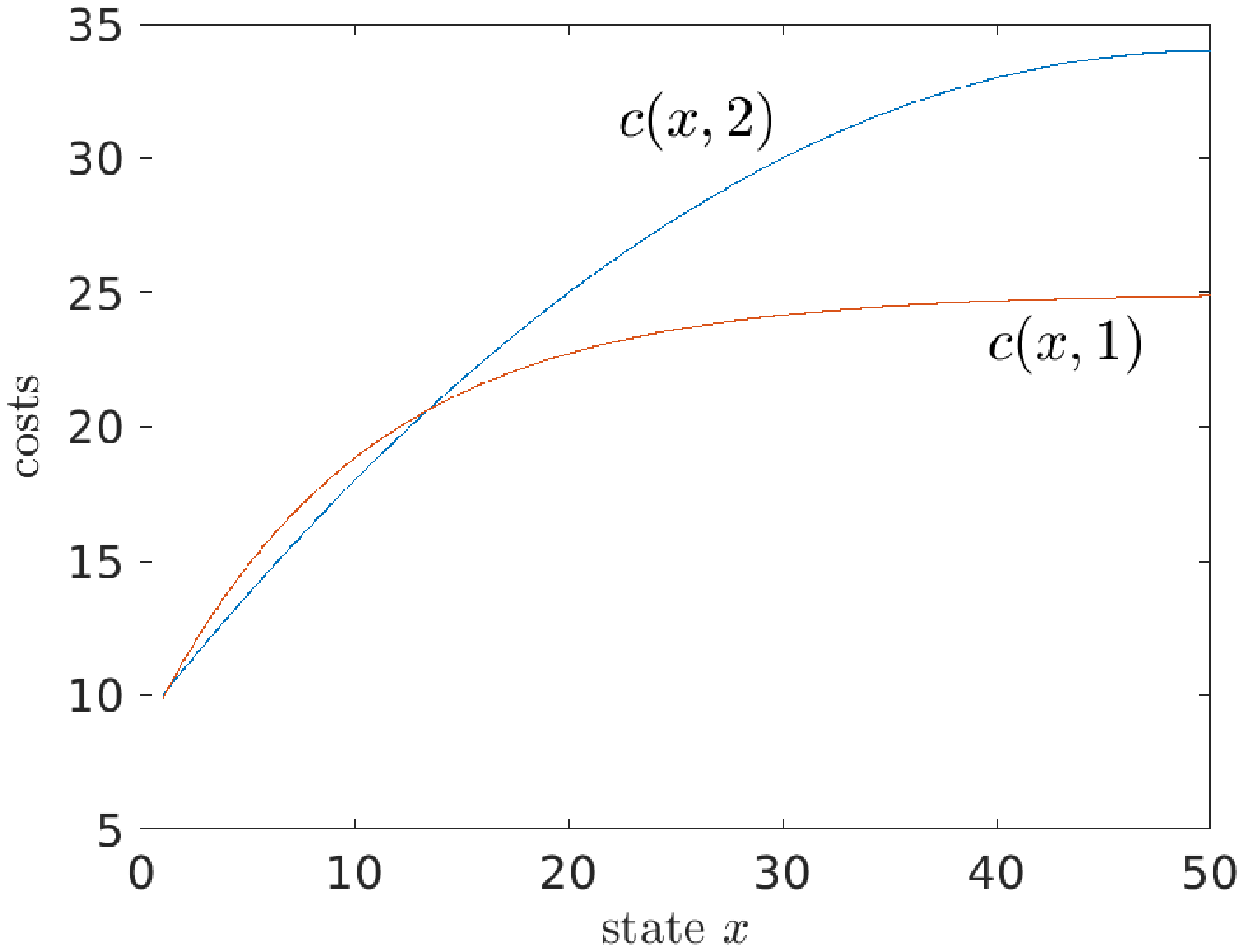}  \hspace{0.3cm}
  \includegraphics[scale=0.4]{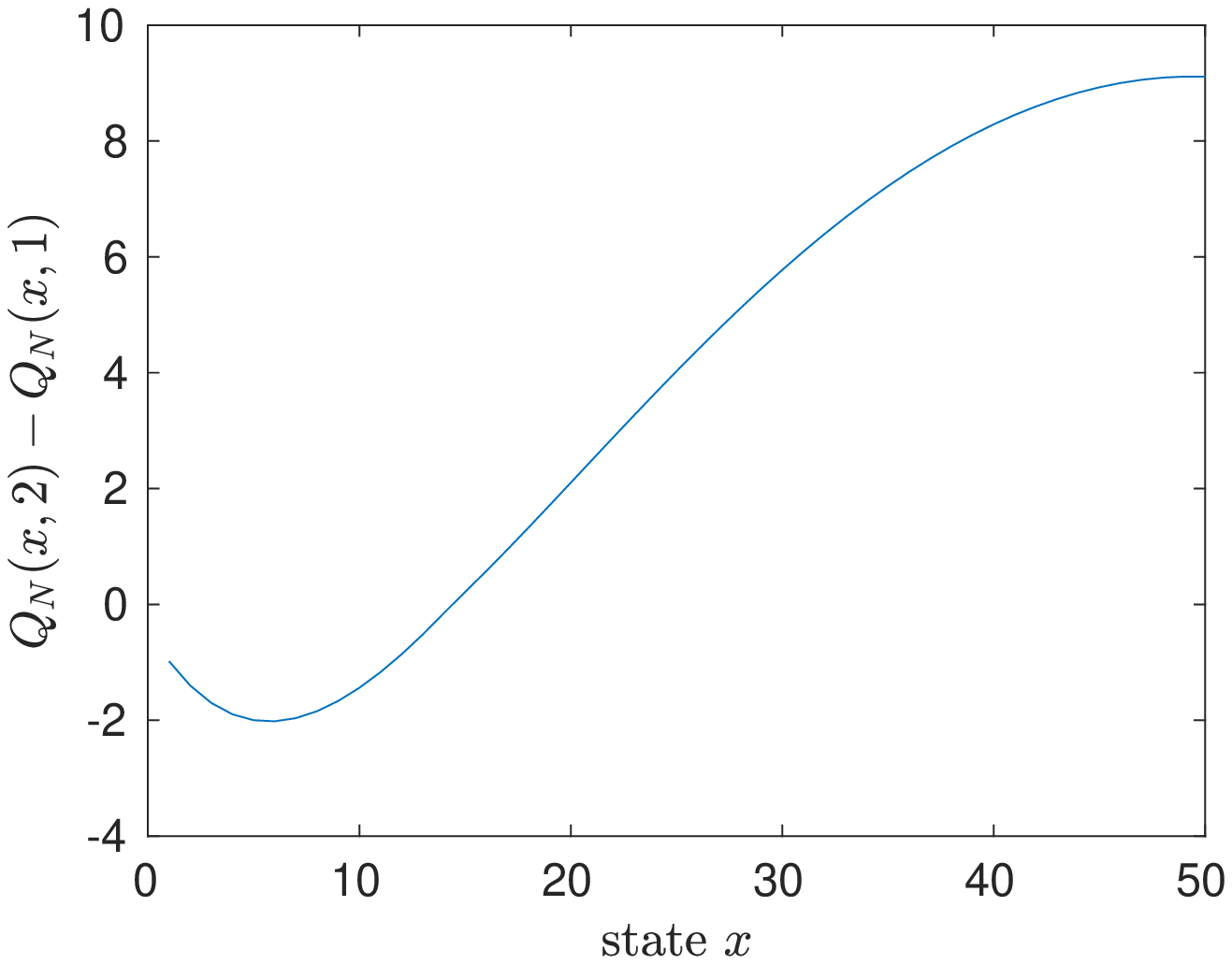} 
  \caption{Concave costs and non-submodular $Q$ function  for MDP with bi-diagonal transition matrix that satisfies Theorem~\ref{thm:convexdom}.}
  \label{fig:concavexmdp}
\end{figure}

\subsubsection*{Example (ii). Tri-diagonal Transition  Matrices and Non-submodular Costs}
 Corollary~\ref{cor:convexdom} applies to  MDPs with tri-diagonal transition matrices where
$
  \tp_{i-1,i}(a) = p_a,\;\tp_{i+1,i} = q_a, \; \tp_{ii} = 1-p_a -q_a, \; \tp_{11}(a) = 1, \; \tp_{\statedim-1,\statedim}=1-s_a, \tp_{\statedim,\statedim}=s_a$.
If $\tp(\action)$ is TP3 and
 $q_a < p_a$, $s_a > 1 + q_a - p_a$ hold, then $\sum_i j \tp_{ij}(a) $ is increasing and convex in $i$; so  modified~\ref{item:TP3} holds.  Also, if $q_a \uparrow a$, $p_a \downarrow a$,  $q_{a+1}-q_a \geq p_{a+1} - p_a$,
$s_{a+1} - s_a > q_{a+1} - q_a+ p_a - p_{a+1}$,
then convex dominance (modified~\ref{item:tpsuperconv}) holds for all $\al \in (0,1]$.
 Then if the costs are chosen so that modified~\ref{convexcost}, and modified~\ref{item:submodc} hold for some
$\be{\pair} \leq 1$, then  Corollary~\ref{cor:convexdom} holds and the optimal policy is monotone (even though the costs are not  submodular when $\be < 1$).

{\em Numerical example}. Consider a discounted cost MDP with $\actiondim=2$, $\statedim = 35$,
tri-diagonal transition  matrices with  $p_1=0.2,p_2 = 0.1, q_1 = 0.05, q_2 =0.1, s_1 = 0.95, s_2 = 1$. Also
 $\discount = 0.95$,
 $\horizon = 200$,
$ \cost(x,1) = -(\tht_1 + \tht_2\,x^3), \quad \cost(x,2) = -(\tht_3 + \tht_4 x^3)$
 where $\tht=[15, 0.3/4^3, 1, 3/4^3]$.
The cost $\cost(x,a)$ is not submodular (see Figure~\ref{fig:tridiag}(i)), but  Corollary~\ref{cor:convexdom} holds. 
Figure~\ref{fig:tridiag}(ii) shows the non-submodular $Q_N(x,a)$.

\begin{figure} \centering
   \includegraphics[scale=0.4]{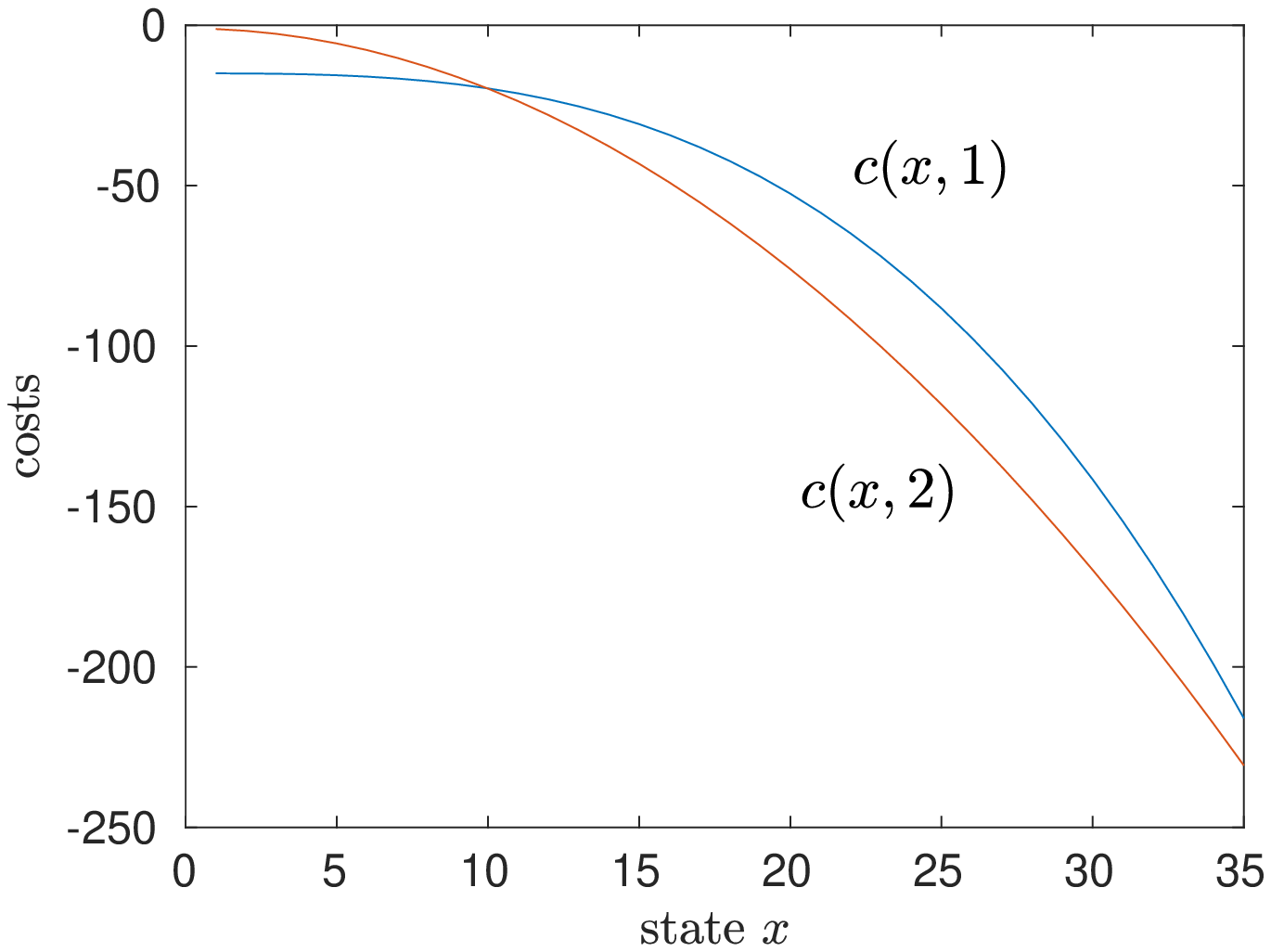}  \hspace{0.3cm}
  \includegraphics[scale=0.4]{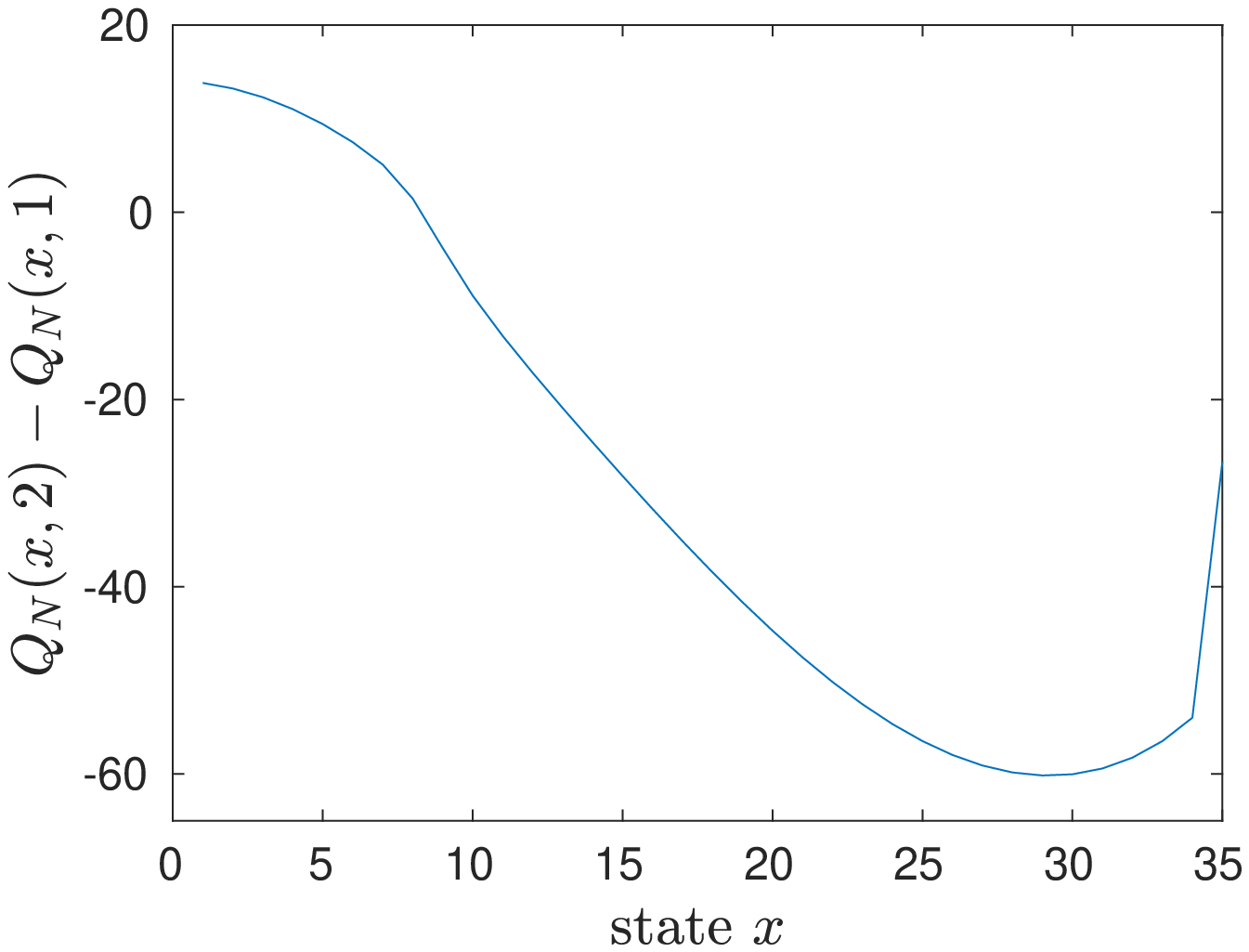} 
  \caption{Non-submodular concave costs  and non-submodular $Q$ function  for MDP with tri-diagonal transition matrix that satisfies Corollary~\ref{cor:convexdom}.}
  \label{fig:tridiag}
\end{figure}

\section{Summary and Discussion} \label{sec:conclusions}

{\em Summary.} The textbook structural result for MDPs uses supermodularity to establish the existence of monotone optimal policies. This paper shows how supermodularity  can be relaxed by formulating a sufficient condition for interval dominance, which we call the \tID condition.   
We presented several examples of MDPs which satisfy \tID including
sigmoidal costs, 
and bi-diagonal/perturbed bi-diagonal transition matrices.    
The structural results in  Sec.\ref{sec:mdp}, namely,
Theorem~\ref{thm:IDmonotone}, Corollaries~\ref{cor:ex0}, \ref{cor:ex2}, \ref{cor:bi1} and Theorem~\ref{thm:ross}
used first order stochastic dominance to establish \tID for several examples. of MDPs. In comparison, Theorem~\ref{thm:convexdom} in Sec.\ref{subsec:concave} discussed examples of \tID in MDPs with concave value functions; we used TP3 assumptions and second order (convex) stochastic dominance to prove the existence of monotone optimal policies.

{\em  Discussion 1. Reinforcement Learning (RL) and Differential sparse Policies:} Once the existence of a monotone optimal policy has been established,  RL  algorithms that exploit this structure can be constructed.
Q-learning algorithms that exploit the \tID condition can be obtained by generalizing the supermodular Q-learning algorithms in~\citet{Kri16}. The second approach is to develop policy search RL algorithms.
In particular, when  $\actiondim$ is small and $\statedim$ is large, then since $\policy^*(\state) \increasing \state$, it is differentially sparse, that is $\policy^*(\state+1) - \policy^*(\state)$ is positive  only at $\actiondim-1$  values of $\state$, and zero for all other $\state$. In~\citet{MRK17}, LASSO based methods are developed to exploit this sparsity and   significantly accelerate search for $\policy^*(\state)$; they build on the nearly-isotonic regression techniques in~\citet{THT11}. The idea is to add a rectified $l_1$-penalty  $\sum_{\state=1}^{\statedim-1} | \policy^l(x) - \policy^l(x+1)|_+$ to the cost in the optimization problem (here $\policy^l$ is the  estimate of the optimal policy at iteration $l$ of the optimization algorithm). Intuitively, this  modifies the cost surface to be more steep in the direction of monotone policies  resulting in faster convergence of an iterative optimization algorithm.

{\em Discussion 2. Convex Value Functions:} Can Theorem~\ref{thm:convexdom} be extended to MDPs with convex value functions? Since convexity is not preserved by minimization,  we  need   multimodularity assumptions  to show  the value function is convex.
However, since multimodularity implies  supermodularity, we are unable to exploit the weaker \tID condition. Multimodularity is sufficient (but not necessary) for convexity to be preserved by minimization; so it is worthwhile exploring relaxed \tID based versions that do not require supermodularity.

\bibliographystyle{abbrvnat}
\bibliography{$HOME/texstuff/styles/bib/vkm}

\begin{thebibliography}{15}
\providecommand{\natexlab}[1]{#1}
\providecommand{\url}[1]{\texttt{#1}}
\expandafter\ifx\csname urlstyle\endcsname\relax
  \providecommand{\doi}[1]{doi: #1}\else
  \providecommand{\doi}{doi: \begingroup \urlstyle{rm}\Url}\fi

\bibitem[Amir(2005)]{Ami05}
R.~Amir.
\newblock Supermodularity and complementarity in economics: An elementary
  survey.
\newblock \emph{Southern Economic Journal}, 71\penalty0 (3):\penalty0 636--660,
  2005.

\bibitem[Derman et~al.(1976)Derman, Lieberman, and Ross]{DLR76}
C.~Derman, G.~J. Lieberman, and S.~M. Ross.
\newblock Optimal system allocations with penalty cost.
\newblock \emph{Management Science}, 23\penalty0 (4):\penalty0 399--403,
  December 1976.

\bibitem[Heyman and Sobel(1984)]{HS84}
D.~P. Heyman and M.~J. Sobel.
\newblock \emph{Stochastic Models in Operations Research}, volume~2.
\newblock McGraw-Hill, 1984.

\bibitem[Kahneman and Tversky(1979)]{KT79}
D.~Kahneman and A.~Tversky.
\newblock Prospect theory: An analysis of decision under risk.
\newblock \emph{Econometrica}, 47\penalty0 (2):\penalty0 263--291, 1979.

\bibitem[Karlin(1968)]{Kar68}
S.~Karlin.
\newblock \emph{Total Positivity}, volume~1.
\newblock Stanford Univ., 1968.

\bibitem[Krishnamurthy(2016)]{Kri16}
V.~Krishnamurthy.
\newblock \emph{Partially Observed Markov Decision Processes. From Filtering to
  Controlled Sensing}.
\newblock Cambridge University Press, 2016.

\bibitem[Lehmann and Casella(1998)]{LC98}
E.~Lehmann and G.~Casella.
\newblock \emph{Theory of Point Estimation}.
\newblock Springer-Verlag, 2nd edition, 1998.

\bibitem[Mattila et~al.(2017)Mattila, Rojas, Krishnamurthy, and
  Wahlberg]{MRK17}
R.~Mattila, C.~R. Rojas, V.~Krishnamurthy, and B.~Wahlberg.
\newblock Computing monotone policies for markov decision processes: a
  nearly-isotonic penalty approach.
\newblock \emph{IFAC-PapersOnLine}, 50\penalty0 (1):\penalty0 8429--8434, 2017.

\bibitem[Milgrom and Shannon(1994)]{MS94}
P.~Milgrom and C.~Shannon.
\newblock Monotone comparative statics.
\newblock \emph{Econometrica}, 62\penalty0 (1):\penalty0 157--180, 1994.

\bibitem[Ngo and Krishnamurthy(2010)]{NK10}
M.~H. Ngo and V.~Krishnamurthy.
\newblock Monotonicity of constrained optimal transmission policies in
  correlated fading channels with {ARQ}.
\newblock \emph{IEEE Transactions on Signal Processing}, 58\penalty0
  (1):\penalty0 438--451, 2010.

\bibitem[Puterman(1994)]{Put94}
M.~Puterman.
\newblock \emph{{M}arkov Decision Processes}.
\newblock John Wiley, 1994.

\bibitem[Quah and Strulovici(2009)]{QS09}
J.~Quah and B.~Strulovici.
\newblock Comparative statics, informativeness, and the interval dominance
  order.
\newblock \emph{Econometrica}, 77\penalty0 (6):\penalty0 1949--1992, 2009.

\bibitem[Ross(1983)]{Ros83}
S.~Ross.
\newblock \emph{Introduction to Stochastic Dynamic Programming}.
\newblock Academic Press, San Diego, California., 1983.

\bibitem[Tibshirani et~al.(2011)Tibshirani, Hoefling, and Tibshirani]{THT11}
R.~J. Tibshirani, H.~Hoefling, and R.~Tibshirani.
\newblock Nearly-isotonic regression.
\newblock \emph{Technometrics}, 53\penalty0 (1):\penalty0 54--61, 2011.

\bibitem[Topkis(1998)]{Top98}
D.~M. Topkis.
\newblock \emph{Supermodularity and Complementarity}.
\newblock Princeton University Press, 1998.

\end{thebibliography}

\begin{appendices}

\section{Toy Example Illustrating \tID Rewards}
Sec.~\ref{sec:ex1} discussed an intuitive visualization of \tID property of rewards for a MDP in terms of the reward curves vs state.
We now provide  a  second intuitive visualization displayed in Figure \ref{fig:ex1b}  in terms of the reward curves vs actions $\action$. This is similar to the discussion in~\citet{QS09}. For  supermodular rewards, $r(3,a)- r(2,a)$  increases with  $a$. Examining   Figure~\ref{fig:ex1b},    for $ a\leq a_x$ (where
  $a_x = \argmax_a r(x,a)=2$), the differences between the reward curves can be arbitrary, as long as  $r(3,a) \geq r(2,a)$ as required by~\ref{cost}.
  Also the $r(3,a)$ curve satisfies that $r(3,1) > \reward(3,3)$, while the $r(2,a)$ curve satisfies $r(2,1) < r(2,3)$; so the single crossing property is violated. Yet condition \tID, namely, ~\ref{cost_id} holds. 

 The following MDP   satisfies Theorem~\ref{thm:IDmonotone}:  $\discount = 0.9$, $\statedim=4$, $\actiondim=3$, $ \horizon=100$,
\begin{multline*} 
  \reward= \begin{bmatrix}
    12 & 4 & 0 \\ 16 & 22 & 18 \\ 22 & 23 & 20 \\ 24 & 28 & 30
  \end{bmatrix}, 
  \tp(1) = \begin{bmatrix}
    0.3 & 0.4 & 0.2 & 0.1 \\  0.2 & 0.4 & 0.2 & 0.2\\  0.2 & 0.4 & 0.1 & 0.3\\  0.2 & 0.3 & 0.1 & 0.4
  \end{bmatrix},    \tp(2) = \begin{bmatrix}
    0.3 & 0.3 & 0.2 & 0.2\\  0.2 & 0.3 & 0.2 & 0.3\\  0.2 & 0.3 & 0.1 & 0.4\\  0.2& 0.2 &0.1 &0.5
  \end{bmatrix},
  \tp(3) =
\begin{bmatrix} 
  0.3 & 0.3 & 0.1 & 0.3\\  0.2 & 0.3 & 0.1 & 0.4\\  0.2 & 0.2 & 0.1 & 0.5\\  0.1 &  0.2 & 0.1 & 0.6
  \end{bmatrix}
  \end{multline*}

  Note that $\reward(2,3) - \reward(2,1) > 0$ while $\reward(3,3) - \reward(3,1) < 0$. Jointly, these  violate both supermodularity and also single crossing.     But the \tID condition~\eqref{eq:id} holds since it  compares action 1 with action 2, and action 2 with action 3. Indeed, \ref{cost_id} holds for
  $\be{\action=1} \leq 1/6$ and $\be{\action=2}= 1$.
Also~\ref{tp_id} holds for all $\al \leq 1$.
Figure~\ref{exii:sigmoidreward} plots $Q_{100}(\state,2) - Q_{100}(\state,1)$ and
$Q_{100}(\state,3) - Q_{100}(\state,1)$. Clearly $Q_k(\state,\action)$ is 
  not supermodular. But the optimal policy is monotone by  Theorem~\ref{thm:IDmonotone}.

  \begin{figure}[h]
    \begin{subfigure}{0.45\linewidth}
  \includegraphics[scale=0.45]{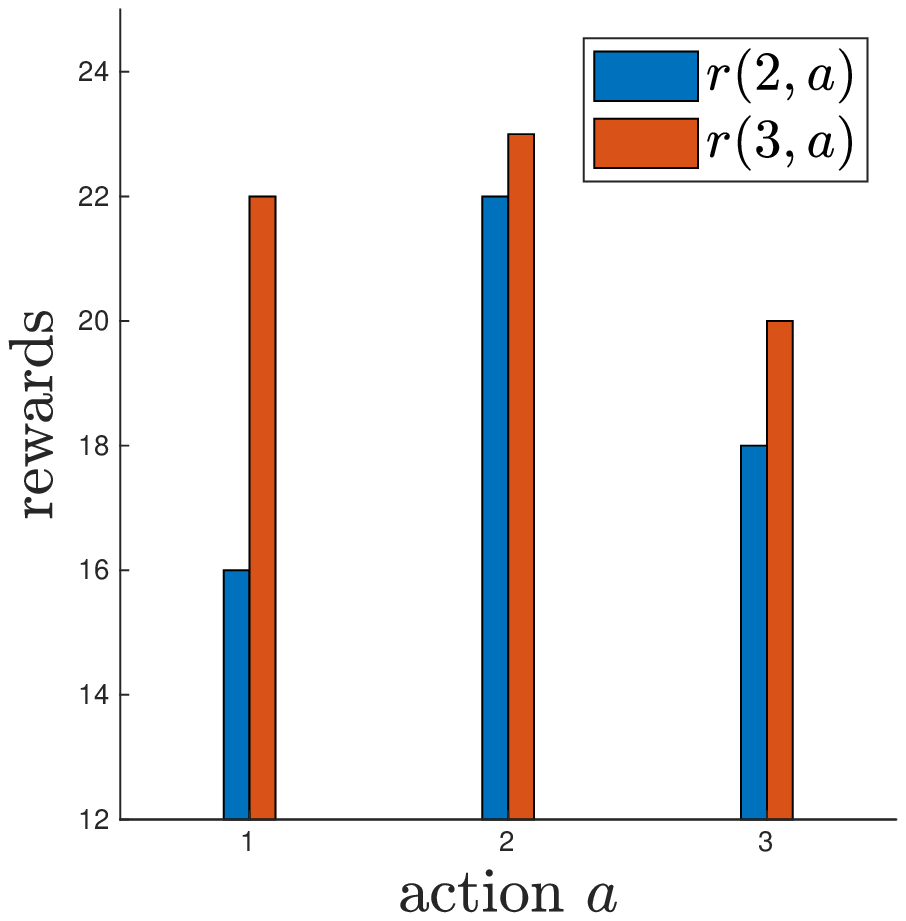}
   \caption{ Non-single crossing  reward} \label{fig:ex1b}
   \end{subfigure}
 \begin{subfigure}{0.45\linewidth}
   \includegraphics[scale=0.4]{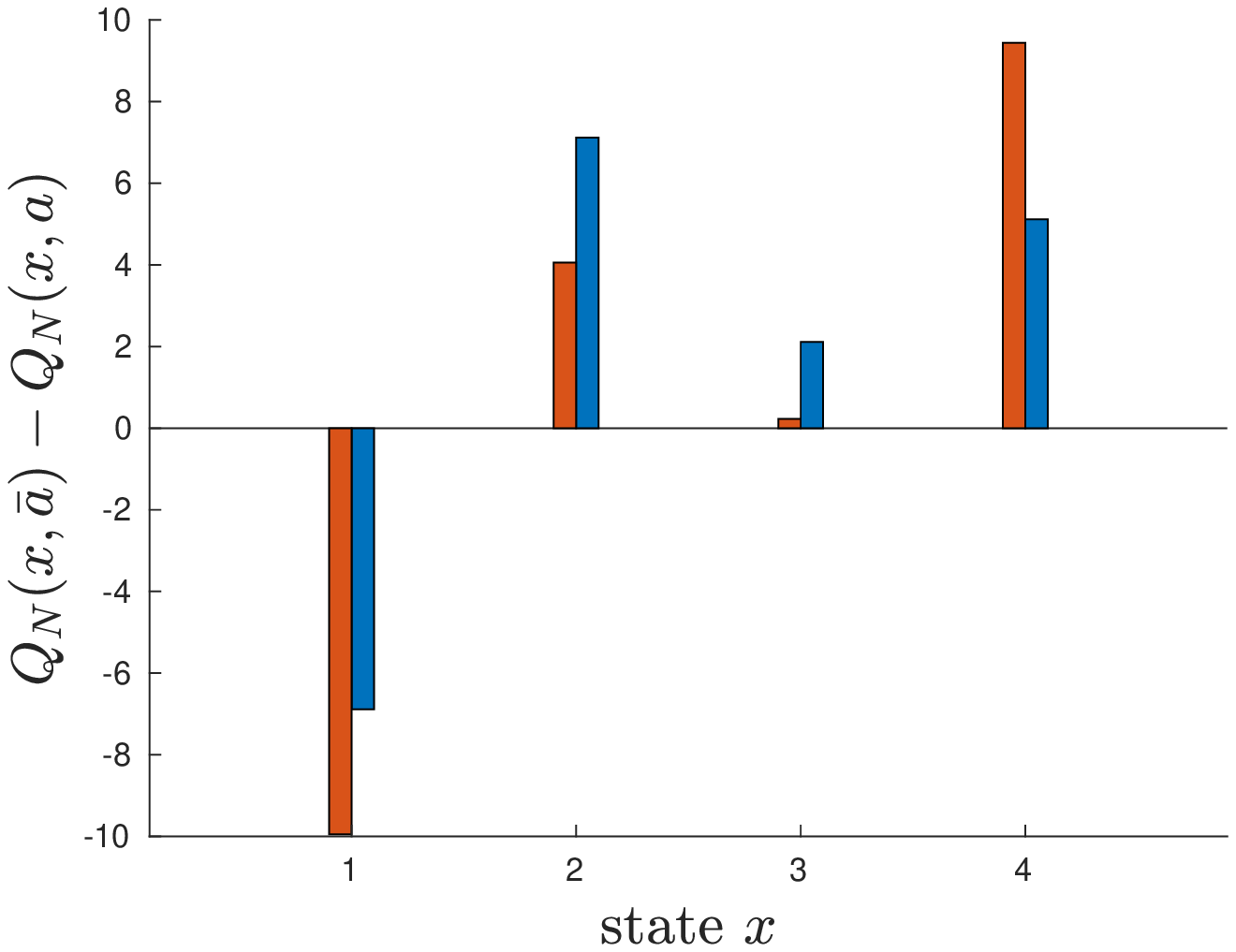}
   \caption{Non-single crossing $Q$ function}
   \label{exii:sigmoidreward}
 \end{subfigure}
    \caption{MDP with  $\statedim=4,\actiondim=3$. If supermodularity holds then  the bars would be increasing with $\state$. Yet \tID holds by Corollary~\ref{cor:ex0} and the optimal policy is monotone.}
  \label{fig:toy}
\end{figure}

\section{Optimal Allocation MDP with Penalty Cost} \label{sec:ross}

This section discusses a finite horizon penalty-cost MDP with perturbed bi-diagonal transition matrices~\eqref{eq:bidiagonale}.   This has applications in  optimal allocation  problems  with penalty costs \citep{Ros83,DLR76} and  wireless transmission control~\citep{NK10}.
We assume $\saram < \param{\action+1}-  \param{\action}$; so as discussed in Sec.~\ref{sec:discbi} supermodularity condition \ref{tp_supermod} does not hold. 

As in  Example 4.2 in \citet[pg.8]{Ros83} and \citet{DLR76}, we consider an $\horizon$-horizon MDP model. There are $\horizon$-stages to construct $\statedim$ components sequentially.
If effort $\cost(\state,\action)$ is allocated then the component is constructed with successfully with probability $\param{a}$.  Our transition matrices are
specified by the perturbed bi-diagonal matrices~\eqref{eq:bidiagonale}.
 At the end of $\horizon$ stages, the penalty cost incurred  is
$\terminal{i}$ if we are $i$ components short, where $i = \{1,\ldots,\statedim\}$, with $\terminal{1}=0$.
\citet{Ros83} considers a continuous action space as  the closed interval $\actionspace=[0,\actiondim]$,  $\cost(x,a) = a$ where $a\in \actionspace$ and bi-diagonal matrices ($\saram=0$).
Although the  \tID condition yields degenerate policies for $\cost(x,a) = a$,
it applies to non-supermodular  cost structures with perturbed bi-diagonal matrices. Such cases  cannot be handled by the convexity based supermodularity approach.

We
consider  the 
discrete  action space $\actionspace = \{1,\ldots,\actiondim\}$
corresponding to  discretization of the continuous valued   actions:
$ \cactionspace = \{0, \; \saram, \; 2\,\saram,\ldots, (\actiondim-1)\,\saram \}$.
Recall $\saram$ are perturbation probabilities of  the bi-diagonal transition matrices in~\eqref{eq:bidiagonale}.
The costs and transition probability parameter $\param{\action}$ in terms of the discretized actions are
\begin{equation}
  \text{ Costs: }  \cost(\state,\action)\,\saram, \quad
  \param{a+1} - \param{a} = \saram\, \dparam\; \; \text{ where } 
  \dparam > 0 .
      \label{eq:MDProssparam}
\end{equation}

  We  make the following assumptions; they are discussed after Theorem~\ref{thm:ross} below.
\begin{enumerate}[label=(A{\arabic*})]
  \setcounter{enumi}{\value{assum_index}}
\item\label{item:monotoneparam} $\dparam\geq  1 $  and    $\uparrow  \action$. (The
  $\uparrow a$ can be relaxed, see remark below.)
\item\label{item:inc}
  Terminal cost  $\terminalc{\state} $ convex and  $\uparrow \state$ with $\terminal{1} = 0$. Cost $\cost(\state,\action) \downarrow \state$.
(More generally,  $\bcost(\state,\action)$  in~\eqref{eq:modifiedval} $\downarrow \state$.)
\end{enumerate}


\noindent {\bf Main Result}.
We will work with the modified value function  $\W_k(\state) = \values_k(\state) - \terminal{\state}$. This is convenient since the  terminal condition  is  $\W_\horizon(i) = 0$ for all $i$.  The dynamic programming recursion~\eqref{eq:finitedp} expressed in terms of  $\W_k(\state)$ and minimizing the cumulative cost  (rather than maximizing the cumulative  reward)  is   
      \begin{align}
\policy^*_{k}(\state) &=   \argmin_\action  \Qw_k(\state,\action) , \quad
    \W_k(\state)= \min_\action  \Qw_k(\state,\action) , \quad k=0,\ldots,\horizon-1 \nn\\
    \Qw_k(i,\action) &=   \bcost(i,\action) + \big(1-\param{a}-\saram(\actiondim-\action)\big)\,\W_{k+1}(i) + \param{a} \W_{k+1}(i-1) \nn \\ \text{ where } & \quad
    \bcost(i,\action)=  \saram\, \cost (i,\action) + \param{\action} \,(\terminalc{i-1} - \terminalc{i}) +
    \saram\,(\actiondim-\action)\,(\terminal{\statedim}-\terminalc{i}), \quad i =1,\ldots,\statedim-1  \label{eq:modifiedval} \\
    &   \Qw_k(\statedim,\action) = \bcost(\statedim,\action) + \param{a} \W_{k+1}(\statedim-1) + (1 - \param{a})\, \W_{k+1}(\statedim), \;
    \bcost(\statedim,\action) = \saram\, \cost(\statedim,a) + \param{a}
  (\terminal{\statedim-1} - \terminal{\statedim}) \nn
  \end{align}


\begin{theorem} \label{thm:ross}
  Consider the $\horizon$-horizon MDP with costs and transition probabilities specified by~\eqref{eq:MDProssparam}, \eqref{eq:bidiagonale}.
  Assume~\ref{item:monotoneparam} and  \ref{item:inc}.
  Suppose $\min_a \dparam> 1$  and the
  costs  satisfy
\begin{equation} \label{eq:maincost}
   \terminal{i+1} \geq \terminal{\statedim} + \frac{\dparam^2\, (\terminal{i} - \terminal{i-1})}{\dparam-1} +
\frac{ \diffcost(i+1,a) - \dparam \diffcost(i,a)}{\dparam-1}
, \quad i = 2,\ldots,\statedim-1
\end{equation}
where  $\diffcost(i,a) = \cost(i,a+1) - \cost(i,a)$ and
 perturbation probabilities  $\saram \in \big(0, \min_a (\param{a+1} - \param{a})\big)$.
   Then  optimal policy $\policy^*_k(i)$,  $k=1,\ldots,\horizon-1$, is increasing in state $i$.
  \end{theorem}


 \noindent  {\bf Remarks}. 
  {\bf 1.}
Theorem~\ref{thm:ross} can be  viewed as complementary result  to the structural result in  \citet{Ros83,DLR76}. On the one hand,
if we choose the same instantaneous cost as~\citet{Ros83}, namely $\cost(\state,\action) = \fv\, a$ for some constant $\fv$,
  then~\eqref{eq:maincost} becomes
$ \terminal{i+1} \geq \terminal{\statedim} + \frac{\dparam^2\, (\terminal{i} - \terminal{i-1})}{\dparam-1} - \fv
$.
But terminal costs satisfying this condition yield monotone policies that are degenerate, namely, $\policy^*_k(i) = 1$ for all~$i$.
So for  $\cost(\state,\action) = \fv\, a$, the \tID condition does not yield a useful result. It is necessary to exploit convexity of the value function, as in \citet{Ros83},  to obtain  non-degenerate  optimal policies.

On the other hand, the \tID condition \eqref{eq:maincost} allows for non-submodular
 costs  and yields  monotone policies (see examples below).
For such cases,  it is not clear  how to extend the convexity based submodularity proof in~ \citet{Ros83} (which applies when $\saram= 0$) to the MDP~\eqref{eq:bidiagonale} for arbitrary  $\saram> 0$.
\\
{\bf 2.} Regarding the assumptions,
\ref{item:monotoneparam} is equivalent to  $\param{\action} \uparrow a$ and convex.  
\ref{item:monotoneparam} can be relaxed to $\param{\action} \uparrow \action$ by imposing stronger conditions on~\eqref{eq:maincost}, see~\eqref{eq:stronger} below.
The  convexity~\ref{item:inc} of terminal costs  implies  $\bcost(i,a)$ in~\eqref{eq:modifiedval} is  decreasing. Recall  decreasing costs~\ref{cost} is used to show submodularity (and Theorem~\ref{thm:IDmonotone}).
\\
{\bf 3.}  Theorem~\ref{thm:ross} considers costs (negative of rewards) whereas Theorem~\ref{thm:IDmonotone} considers rewards. Note \ref{cost} is equivalent to the   cost decreasing in the state. Also   inequality~\ref{cost_id} is reversed   in terms of costs.

\begin{figure}[t]
  \centering
  \begin{subfigure}{0.45\linewidth}
    \includegraphics[scale=0.4]{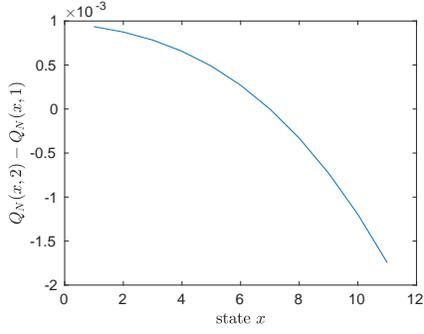}
    \caption{$Q$ function for submodular cost}
  \end{subfigure} \hspace{1cm}
  \begin{subfigure}{0.45\linewidth}
    \includegraphics[scale=0.4]{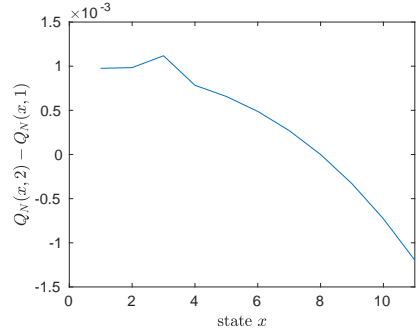}
    \caption{$Q$ function for non-submodular cost}
  \end{subfigure}
  \caption{$Q$ function for   Optimal Allocation MDP}
  \label{fig:rossmdp}
\end{figure}

\noindent {\bf Examples}. We chose the MDP parameters in~\eqref{eq:bidiagonale}, \eqref{eq:MDProssparam}
as $\statedim= 11$, $\actiondim=2$, $\dparam=1.2$, $\saram=10^{-6}$, 
$\terminal = [0,1,2,4, 8 , 15, 25, 40, 60, 90,200]$. Figure~\ref{fig:rossmdp} displays $\Q_k(x,2) - \Q_k(x,1)$ for  two cases: (i) 
  $\cost(x,1) = 0$, $\cost(x,2) = \saram (\fc - (x+2)^3)$, $\fv =10^3$  (ii) 
  $\cost(x,1) = 0$, $\cost(x,2) = \saram (\fc + 2.5\, x^4 I(x \leq3) - (x+2)^3)$,  $\fv = 10^3$. \\
  In case (ii), $Q(\state,\action)$ is not submodular; see Figure~\ref{fig:rossmdp}(b).
  But  Theorem~\ref{thm:ross} holds;  so optimal policy $\policy_k^*(\state) \uparrow x$.
  
{\bf Proof of Theorem~\ref{thm:ross}}
Using the modified dynamic programming recursion~\eqref{eq:modifiedval}, we   verify  the assumptions in Theorem~\ref{thm:IDmonotone}.
Examining the transformed cost $\bcost(\state,\action) $ in~\eqref{eq:modifiedval}, \ref{cost} holds for $\bcost(\state,\action)$,  $\state\in \{1,\ldots,\statedim\}$, if $\terminal{i}$ is convex and increasing, and $\cost(i,a)$ is decreasing in $i$, i.e,
  \ref{item:inc} holds.
  From the structure of $\tpe(\action)$ in~\eqref{eq:bidiagonale}, \ref{tps} holds.
The terminal cost  in ~\eqref{eq:modifiedval} is $0$ for all states;  so~\ref{terminal} holds trivially.
Next by~\ref{item:monotoneparam}, $\da{\action} > 0$. So for  actions $\action$ and $\action+1$,
it is easily verified that~\ref{tp_id} holds for $\al{\action} \geq \da{a}/\saram > 1$.
So $\saram \in (0,\min_\action \da{a}]$.

We  now establish~\ref{cost_id}  for $\bcost(\state,\action)$,
$\state \in \{1,\ldots,\statedim-1\}$. 
Choose
  $\be{\action} = \al{\action} =  \da{a}/\saram  = \dparam$
  (substituting~\eqref{eq:MDProssparam} for $\da{a}$).  By~\ref{item:monotoneparam}, $\al{\action} > 0 \uparrow a$. 
  Then~\ref{cost_id} is equivalent to
 $\saram  (\cost(\state+1,a+1) - \cost(\state+1,a)) + \saram\, \dparam\,(\terminal{\state}- \terminal{\state+1}) + \saram \,(\terminal{\state+1} - \terminal{\statedim})  \leq
    \be{\action}\, [ \saram  (\cost(\state,a+1) - \cost(\state,a)) + \saram\, \dparam\,(\terminal{\state-1}- \terminal{\state}) + \saram \,(\terminal{\state} - \terminal{\statedim})]$.
The positive  parameter $\saram$ cancels  on both sides, yielding~\eqref{eq:maincost}.
Since $\be{a} = \al{a}$,  \ref{sum} holds. Thus all conditions of Theorem~\ref{thm:IDmonotone} hold.

{\em Remark}. Choosing $\al =\bgam =  \max_a \dparam$ in the  proof, we obtain a stronger sufficient condition than~\eqref{eq:maincost}:
\begin{equation}
  \label{eq:stronger}
  \terminal{i+1} \geq \terminal{\statedim} \frac{\bgam-1}{\dparam-1}
  + \frac{\bgam\,\dparam\, (\terminal{i} - \terminal{i-1})}{\dparam-1}
  + \frac{ \diffcost(i+1,a) - \bgam \diffcost(i,a)}{\dparam-1} + \frac{(\dparam - \bgam)\terminal{i}}{\dparam-1}
\end{equation}
Since $\al = \be$ is a constant and  not $\action$ dependent, 
~\ref{item:monotoneparam} is relaxed to $\dparam > 1$. 

\end{appendices}

\end{document}